%% file: main.tex
\documentclass[runningheads]{llncs}
\usepackage[crop]{llncsconf}
\usepackage{graphicx}
\usepackage{color}

\newcommand{\msize}[1]{{\left|#1\right|}}

\newcommand{\calP}{{\cal P}}

\newcommand{\ptitle}[1]{\medskip\noindent\hspace*{2mm}\textsc{\underline{#1}}} 

\newtheorem{obs}[theorem]{Observation}

\newenvironment{listing}[1]{%
        \begin{list}{*}{%
                 \settowidth{\labelwidth}{#1}%
                 \setlength{\leftmargin}{\labelwidth}%
                 \advance \leftmargin by 12pt
                   \setlength{\itemsep}{0pt}%
                   \setlength{\parsep}{0pt}%
                   \setlength{\topsep}{0pt}%
                   \setlength{\parskip}{0pt}%
}%
}{%
\end{list}}

\begin{document}
%
\title{Computational Complexity of Jumping Block Puzzles}
%
%
\author{Masaaki Kanzaki\inst{1} \and
  Yota Otachi\inst{2}\orcidID{0000-0002-0087-853X} \and
  Ryuhei Uehara\inst{1}\orcidID{0000-0003-0895-3765}}
\authorrunning{M. Kanzaki et al.}
%
\institute{
School of Information Science, Japan Advanced Institute of Science and Technology.
Asahidai, Nomi, Ishikawa 923-1292, Japan \email{\{kanzaki,uehara\}@jaist.ac.jp}
\and
Department of Mathematical Informatics, Graduate School of Informatics, Nagoya University.
Furocho, Chikusa-ku, Nagoya, Aichi 464-8601, Japan
\email{otachi@nagoya-u.jp}}
\maketitle              
\begin{abstract}
  In combinatorial reconfiguration, the reconfiguration problems on
  a vertex subset (e.g., an independent set) are well investigated.
  In these problems, 
 some tokens are placed on a subset of vertices of the graph,
 and there are three natural
  reconfiguration rules called ``token sliding,'' ``token jumping,'' and ``token addition and removal''.
  In the context of computational complexity of puzzles, the sliding block puzzles play an important role.
  Depending on the rules and set of pieces, the sliding block puzzles 
 characterize the computational complexity classes including P, NP, and PSPACE.
  The sliding block puzzles correspond to the token sliding model in the context of combinatorial reconfiguration.
  On the other hand, a relatively new notion of jumping block puzzles is proposed in puzzle society.
 This is the counterpart to the token jumping model of the
  combinatorial reconfiguration problems in the context of block puzzles.
  We investigate several variants of jumping block puzzles and determine their computational complexities.

\end{abstract}
\section{Introduction}
\label{sec:intro}
\input{intro}

\section{Preliminaries}
\label{sec:pre}
\input{pre}

\section{Flip over puzzles}
\label{sec:flip}
In this section, we focus on the flip over puzzles.
We first show that this puzzle is PSPACE-complete even on quite
restricted input. Next, we show that it is NP-hard even if the frame
$F$ is a rectangle of width 2. Lastly, we show that
it is NP-complete if the number of flips of each piece is $O(1)$.

\subsection{PSPACE-completeness in general}
\label{sec:pspace}
\input{pspace}

\subsection{NP-hardness on a frame of width 2}
\input{np}

\subsection{NP-completeness with constant flips}
\input{np2}

\section{Flying block puzzles}
\label{sec:fly}
\input{fly}

\section{Concluding remarks}
\input{conc}

\section*{Acknowledgments}
A part of this research is supported by JSPS KAKENHI Grant Numbers
JP21K11752, 
JP20K11673, 
JP20H05964, 
JP20H05793, 
JP18K11169, 
JP18K11168, 
JP18H04091, 
and
JP17H06287. 

\bibliographystyle{splncs04}
\bibliography{main}

\end{document}

%% file: intro.tex
Recently, the \emph{reconfiguration problems} attracted the attention in theoretical computer science.
These problems arise when we need to find a step-by-step transformation between two feasible solutions of 
a problem such that all intermediate results are also feasible and each step abides by a fixed reconfiguration rule, 
that is, an adjacency relation defined on feasible solutions of the original problem.
The reconfiguration problems have been studied extensively for several well-known problems, including 
{\sc independent set}~\cite{IDHPSUU,KaminskiMedvedevMilanic2012}, 
{\sc satisfiability}~\cite{Kolaitis,MTY11}, {\sc set cover}, {\sc clique}, {\sc matching}~\cite{IDHPSUU}, 
{\sc vertex coloring}~\cite{BJLPP11,BC09,CHJ11}, and {\sc shortest path}~\cite{KMP11}.

In the reconfiguration problems that consist in transforming a vertex subset (e.g., an independent set),
there are three natural reconfiguration rules called
``token sliding,'' ``token jumping,'' and ``token addition and removal'' \cite{KaminskiMedvedevMilanic2012}.
In these rules, a vertex subset is represented by the set of \emph{tokens} placed on
each of the vertices in the set. 
In the token sliding model, we can slide a token on a vertex to another along an edge joining these vertices.
On the other hand, a token can jump to any other vertex in the token jumping model.

\begin{figure}\centering
\includegraphics[bb=0 0 570 645,width=0.2\linewidth]{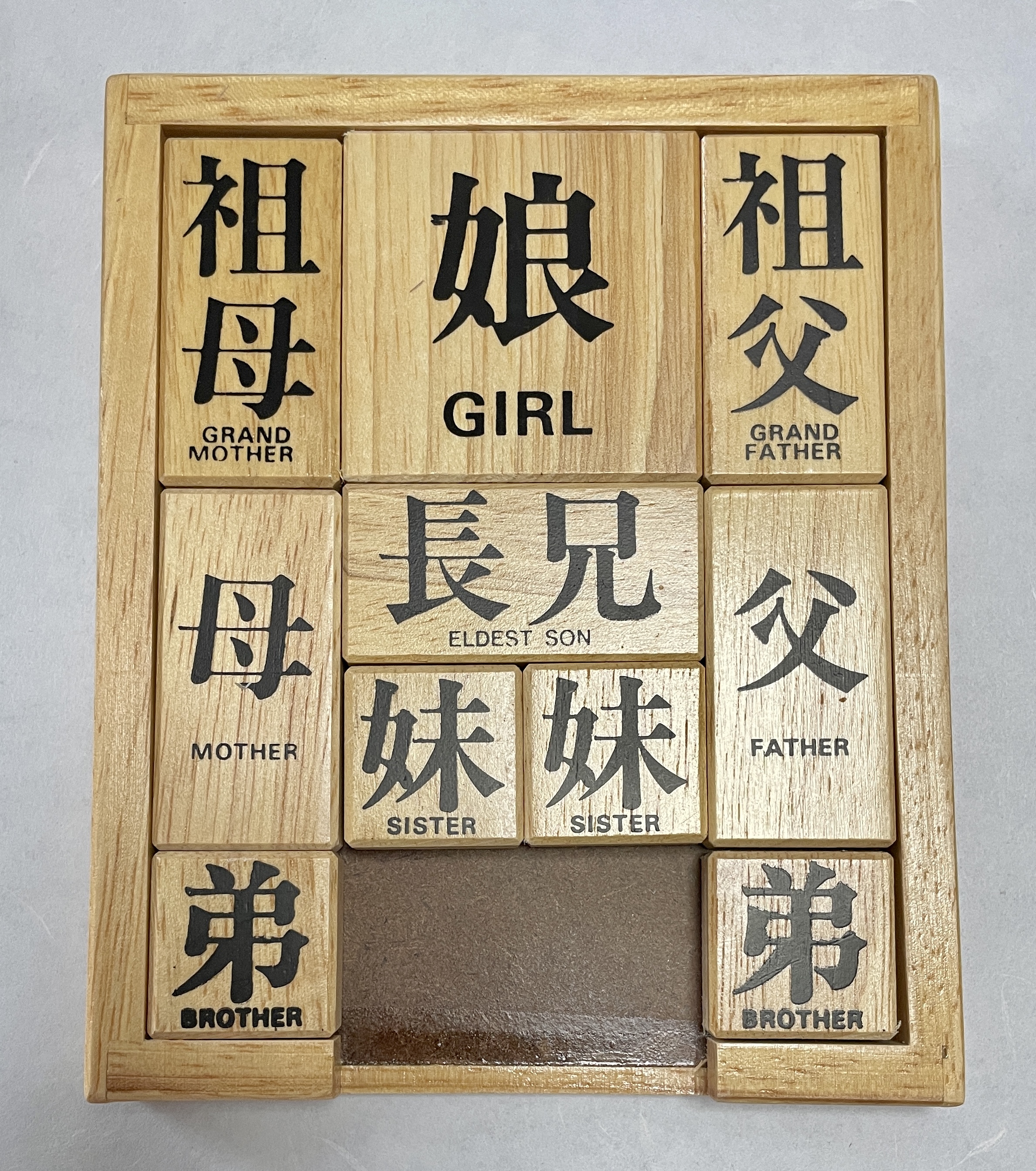}\hspace{1cm}
\includegraphics[bb=0 0 571 553,width=0.25\linewidth]{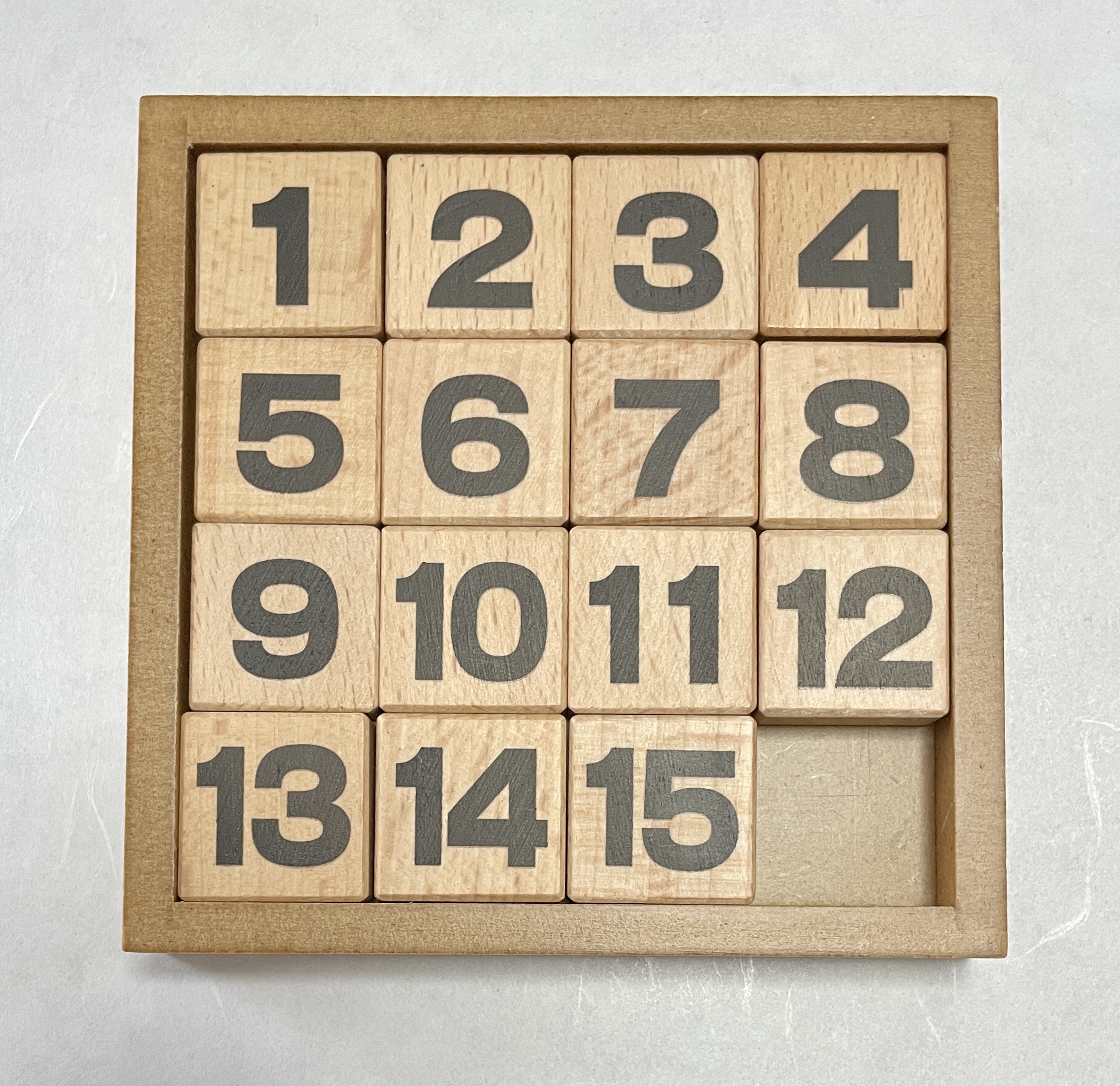}
\caption{A typical sliding block puzzle and the 15-puzzle.}
\label{fig:dad}
\end{figure}

In the puzzle society, tons of puzzles have been invented which can be seen as realizations of some reconfiguration problems.
Among them, the family of \emph{sliding block puzzles} (\figurename~\ref{fig:dad})
has been playing an important role bridging recreational mathematics and theoretical computer science.
A classic puzzle is called the Dad puzzle; it consists of rectangle pieces in a rectangle frame, and
the goal is to slide a specific piece (e.g., the largest one) to the goal position.
The Dad puzzle was invented in 1909 and the computational complexity of this puzzle was open since Martin Gardner mentioned it.
After almost four decades,
Hearn and Demaine proved that the puzzle is PSPACE-complete in general (a comprehensive survey can be found in \cite{HearnDemaine2009}).
When all pieces are unit squares, we obtain another famous puzzle called the \emph{15-puzzle}; in this puzzle,
we slide each unit square using an empty area and arrange the pieces in order.
From the viewpoint of combinatorial reconfiguration, this puzzle has remarkable properties in the general form of the $(n^2-1)$-puzzle.
For given initial and final arrangements of pieces, the decision problem asks 
if we can transform from the initial arrangement to the final arrangement.
Then the decision problem can be solved in linear time, and for a yes-instance, 
while a feasible solution can be found in $O(n^3)$ time, 
finding a shortest solution is NP-complete (see \cite{RW90,DemaineRudoy2018}).

To see how the computational complexity of a reconfiguration problem depends on its reconfiguration rule,
it is natural to consider the ``jumping block'' variant as a counterpart of sliding block puzzles.
In fact, in the puzzle society, some realizations of \emph{jumping block puzzles} have been invented.
As far as the authors know, ``Flying Block'' was the first jumping block puzzle, which was designed by Dries de Clercq
and popularized at International Puzzle Party by Dirk Weber in 2008.\footnote{You can find
``Flying Block'', ``Flying Block II'', and ``Flying Block III'' at \url{http://www.robspuzzlepage.com/sliding.htm} (accessed in June 2021).}
The Flying Block consists of four polyominoes (see \cite{Gol} for the notion of polyominoes) within a rectangle frame.
The goal is similar to the Dad puzzle; moving a specific piece to the goal position.
In one move, we first pick up a piece, rotate it if necessary, and then put it back into the frame.
A key feature is that each piece has a small tab for picking up.
Because of this feature, flipping a piece is inhibited when we put it into the frame.
In this framework, several puzzles were designed by Hideyuki Ueyama,
and some of them can be found in \cite{Akiyama} (see Concluding Remarks for further details).

\begin{figure}
\centering
\includegraphics[width=0.75\linewidth]{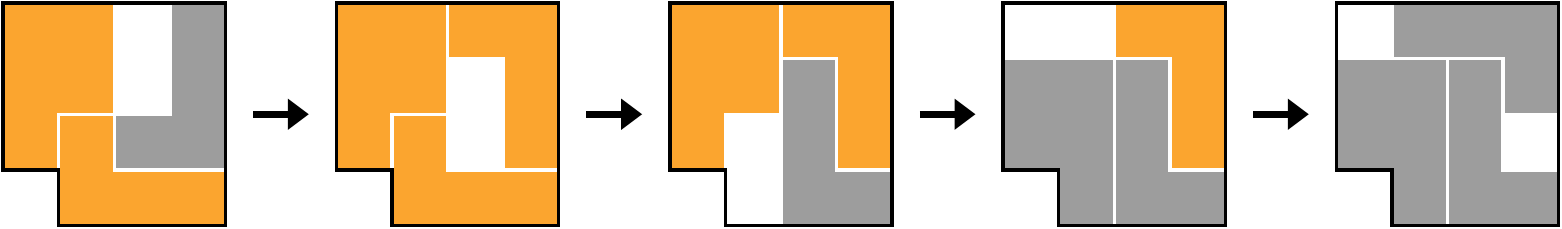}
\caption{An example of the flip over puzzle. Each piece has its front side and back side, and the goal is to make them back-side up.}
\label{fig:sample}
\end{figure}

Later, Fujio Adachi invented another jumping block puzzle
which is called \emph{Flip Over} (or Turn Over). It was invented in 2016
and popularized at International Puzzle Party by Naoyuki Iwase in 2019.\footnote{Some commercial products
can be found at \url{http://www.puzzlein.com/} (in Japanese; accessed in June 2021).}
This puzzle consists of four polyominoes in an orthogonal frame, which is not a rectangle.
The operation is a bit different from Flying Block;
you can \emph{flip} the piece in addition to the rotation if you like.
Each piece has its front side and back side (distinguished by their colors).
The goal of this puzzle is to make all pieces back-side up.
A simple example is given in \figurename~\ref{fig:sample}.

As far as the authors know, the jumping block puzzles have never been investigated in
the context of the reconfiguration problems.
In this paper, we first investigate the jumping block puzzle under the model of Flip Over.
That is, the \emph{flip over puzzle} problem is formalized as follows:

\ptitle{Flip over puzzle}
\begin{listing}{aaa}
\item[{\bf Input:}] A set of polyominoes in an orthogonal frame.
\item[{\bf Operation:}] Pick up a piece, rotate and flip it if desired, and put it back into the frame.
\item[{\bf Goal:}] To make every piece back-side up in the frame.
\end{listing}
\medskip

We first observe that each piece can be flipped in-place if it is line symmetric.
That is, when all polyominoes are line symmetric,
it is a yes-instance and all pieces can be flipped in a trivial (shortest) way.
In contrast with that, the existence of one asymmetric piece changes the computational complexity of this puzzle.
The first result in this paper is the following theorem:
\begin{theorem}
\label{th:main}
The flip over puzzle is PSPACE-complete even with all the following conditions:
(1) the frame is a rectangle without holes,
(2) every line symmetric piece is a rectangle of size $1\times k$ with $1\le k\le 3$,
(3) there is only one asymmetric piece of size 4, and
(4) there is only one vacant unit square.
\end{theorem}
We note that all polyominoes of size at most 3 are line symmetric 
(i.e., monomino, domino, I-tromino, and L-tromino in terms of polyomino).
Therefore, Theorem \ref{th:main} is tight since it contains only one minimum asymmetric piece.

We show Theorem \ref{th:main} by a reduction from the Nondeterministic Constraint Logic (NCL).
It is known that the NCL is PSPACE-complete even if the given input NCL graph has constant bandwidth \cite{Zanden}.
Using the result in \cite{Zanden}, 
we will have the claim in Theorem \ref{th:main} even if the frame is a rectangle of constant width.
On the other hand, if the frame is a rectangle of size $1\times m$, the problem is trivial:
each piece should be a rectangle of width 1 and hence it can be flipped to back side in-place in one step.
Interestingly, the problem is intractable when the frame has width 2.
\begin{theorem}
\label{th:2np}
The flip over puzzle is NP-hard even with all the following conditions:
(1) the frame is a rectangle of width $2$ without holes, 
(2) every line symmetric piece is a rectangle of width 1, and 
(3) there is only one asymmetric piece.
\end{theorem}

In our reduction for proving Theorem \ref{th:2np},
the constructed instances admit a sequence of flips, if any exists, of length polynomial in $n$.
However, we do not know whether the flip over puzzle on a frame of width $w$ is in NP or not for some small constant $w \ge 2$.
On the other hand, when the number of flips of each piece is bounded above by a constant,
we can observe that this problem is in NP.
Under this assumption, we show that the flip over puzzle is NP-complete even if we have some combination of natural conditions.

Next we turn to the \emph{flying block puzzle} problem, which is formalized as follows:

\ptitle{Flying block puzzle}
\begin{listing}{aaa}
\item[{\bf Input:}] A set of polyominoes in an orthogonal frame, a specific piece $P$ in the set, and a goal position of $P$ in the frame.
\item[{\bf Operation:}] Pick up a piece, rotate it if desired, and put it back into the frame.
\item[{\bf Goal:}] To move $P$ to the goal position.
\end{listing}
\medskip

In our results of the flip over puzzle, the unique asymmetric piece plays an important role
and flips of rectangles are not essential.
In some cases, we can show corresponding results for the flying block puzzle 
by modifying the unique asymmetric piece in the flip over puzzle.
Using the idea, we first show natural counterparts of Theorems \ref{th:main} and \ref{th:2np}.
However, while the goal of the flip over puzzle is to flip \emph{all} pieces,
the goal of the flying block puzzle is to arrange the \emph{specific} piece to the goal position.
This difference requires different techniques, and some counterparts of
the NP-completeness results of the flip over puzzle remain open.
Intuitively, throughout these counterparts, the flying block puzzle
seems to be more difficult than the flip over puzzle since the hardness results hold under stronger restrictions.
In fact, the computational complexity of the flying block puzzle is not trivial even if the frame is a rectangle of width 1,
while it is trivial in the flip over puzzle since all pieces are rectangles of width 1.
We show weakly NP-completeness of the flying block puzzle even if the frame is of width 1 and each piece can be
moved at most once. On the other hand, we show a nontrivial polynomial time algorithm
when we can move each piece any number of times in the frame of width 1.


%% file: pre.tex
A \emph{polyomino} is a polygon formed by joining one or more unit squares edge to edge.
A polyomino is also called a \emph{$k$-omino} if it consists of $k$ unit squares.
When $k=1,2,3,4$, we sometimes call it \emph{monomino}, \emph{domino}, \emph{tromino}, and \emph{tetromino}, respectively.
An instance of the \emph{jumping block puzzle} consists of a set
$\calP$ of polyominoes $P_1,P_2,\ldots,P_n$ and a polyomino $F$.
Each polyomino $P_i\in \calP$ is called a \emph{piece}, and $F$ is called a \emph{frame}.
Each piece $P_i$ has its \emph{front side} and \emph{back side}
(in the figures in this paper, a bright color and a dark color indicate front side and back side, respectively).

A \emph{feasible packing} of $\calP$ to $F$ is an arrangement of all pieces of $\calP$ into $F$ such that
each piece is placed in $F$, no pair of pieces overlaps (except at edges and vertices),
every vertex of pieces of $\calP$ is placed on a grid point in $F$,
and each edge of pieces of $\calP$ is parallel or perpendicular to the edges of $F$.

For a feasible packing of $\calP$ to $F$, a \emph{flip} of $P_i$ is an operation that consists of the following steps:
(1) pick up $P_i$ from $F$, (2) translate, rotate, and flip $P_i$ if necessary, and
(3) put $P_i$ back into $F$ so that the resulting arrangement is
a feasible packing.
On the other hand, a \emph{fly} of $P_i$ is an operation that consists of the following steps:
(1) pick up $P_i$ from $F$, (2) translate and rotate $P_i$  if necessary, and
(3) put $P_i$ back into $F$ so that the resulting arrangement is a feasible packing.

For a given feasible packing $X$ of $\calP$ to $F$, the \emph{flip over puzzle} asks
whether $X$ can be reconfigured by a sequence of flips to a feasible packing in which
all pieces are back-side up in $F$.
In contrast, in the flying block puzzle, the input consists of three tuples;
a feasible packing of $\calP$ to $F$, a specific piece, say $P_n$, and the goal position of $P_n$ in $F$.
That is, the \emph{flying block puzzle} asks if
we can move the specific piece $P_n$ to the goal position starting from the feasible packing by a sequence of flies.

In order to determine the computational complexities of the jumping block puzzles,
we use the notion of the \emph{constraint logic} which was introduced 
by Hearn and Demaine \cite{HearnDemaine2005} (see also \cite{HearnDemaine2009}).
A \emph{constraint graph} $G=(V,E,w)$ is an edge-weighted undirected 3-regular graph such that
(1) each edge $e\in E$ is weighted $1$ or $2$,
(we sometimes describe the values 1 by \emph{red} and 2 by \emph{blue}, respectively) and
(2) each vertex is either an \emph{AND vertex} or an \emph{OR vertex} such that
(2a) an AND vertex is incident to two red edges and one blue edge, and
(2b) an OR vertex is incident to three blue edges.
A \emph{configuration} of a constraint graph is an orientation of the edges in the graph.
A configuration is \emph{legal} if the total weight of the edges pointing to each vertex is at least 2.
The problem \emph{NCL} on the constraint logic is defined as follows.

\newpage 

\ptitle{NCL}
\begin{listing}{aaa}
\item[{\bf Input:}] A constraint graph $G=(V,E,w)$, a legal configuration $C_0$ for $G$, and an edge $e_t\in E$.
\item[{\bf Question:}] Is there a sequence of legal configurations $(C_0,C_1,\ldots)$ such that
(1) $C_i$ is obtained by reversing the direction of a single edge in $C_{i-1}$ with $i\ge 1$, and
(2) the last configuration is obtained by reversing the direction of $e_t$?
\end{listing}
\medskip

The \emph{Bounded NCL} is a variant of the NCL that requires one additional restriction that every edge can be reversed at most once.
For these two problems, the following theorem is known:

\begin{theorem}[\cite{HearnDemaine2009}]
\label{th:NCL}
(1) NCL is PSPACE-complete even on planar graphs and
(2) Bounded NCL is NP-complete even on planar graphs.
\end{theorem}

%% file: pspace.tex
\begin{proof}
We give a polynomial-time reduction from the NCL problem to the flip over puzzle.
Let $G=(V,E,w)$, $e_t$, and $C_0$ be the input of the NCL problem.
That is, $G$ is an edge-weighted planar graph and $C_0$ is a legal configuration of $G$ and $e_t$ is an edge in $E$.
The framework of the reduction is similar to the reduction of the NCL
problem to the sliding block puzzle shown in \cite[Sec.~9.3]{HearnDemaine2009}:
The frame $F$ is a big rectangle, and it is filled with a regular grid of gate gadgets,
within a ``cell border'' construction. The internal construction of the gates is such that
none of the cell-border blocks may move, thus providing overall integrity to the configuration.
In our reduction, the cell border has width 2 (\figurename~\ref{fig:border}), 
and each gadget in a cell will be designed of size $13\times 13$.
We first prepare a sufficiently large cell border for embedding $G$.
Then we construct an instance of the flip over puzzle in three steps.

\begin{figure}\centering
\includegraphics[width=0.5\linewidth]{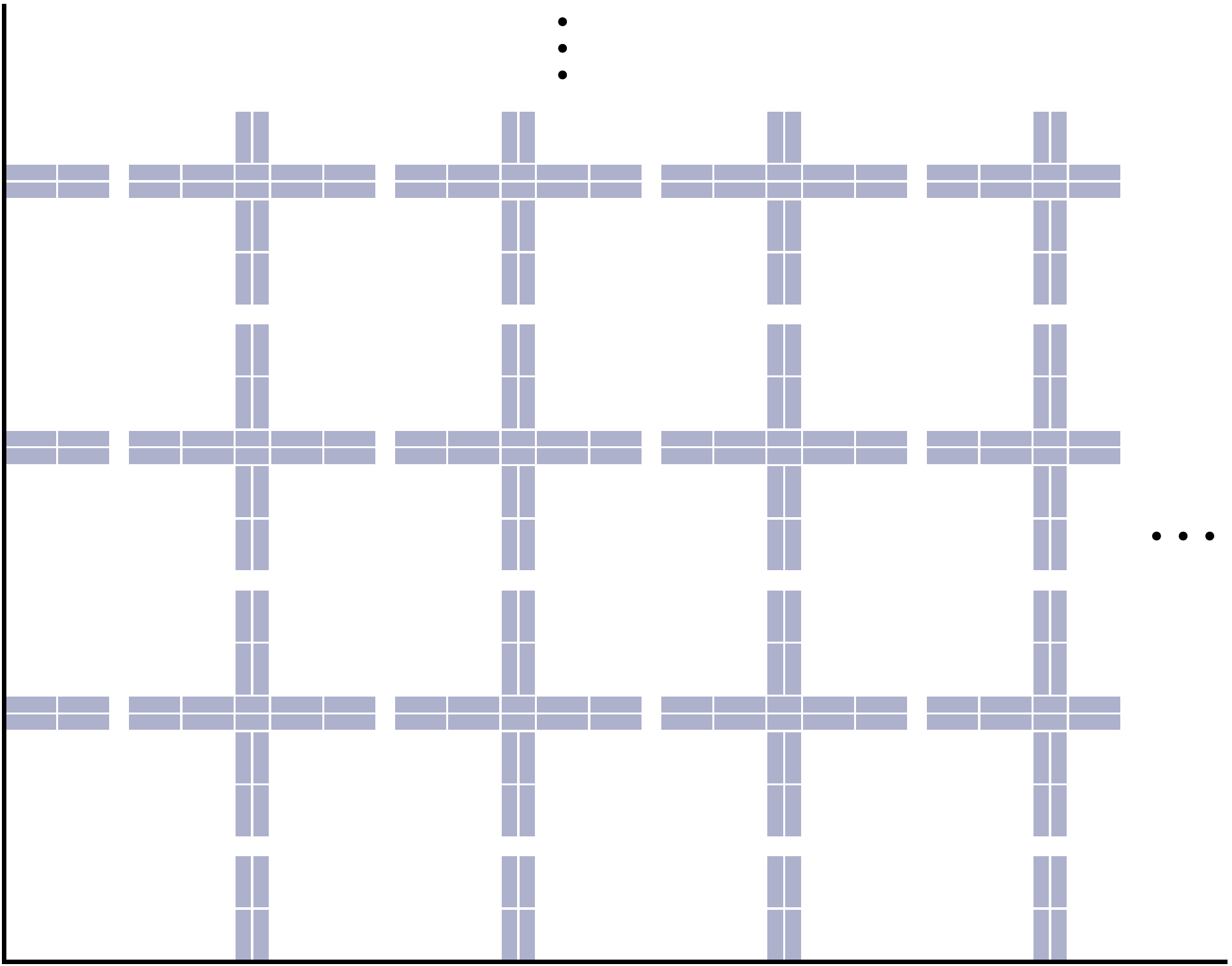}
\caption{The cell border, which is a framework of embedding gadgets.}
\label{fig:border}
\end{figure}

\begin{figure}\centering
\includegraphics[width=0.8\linewidth]{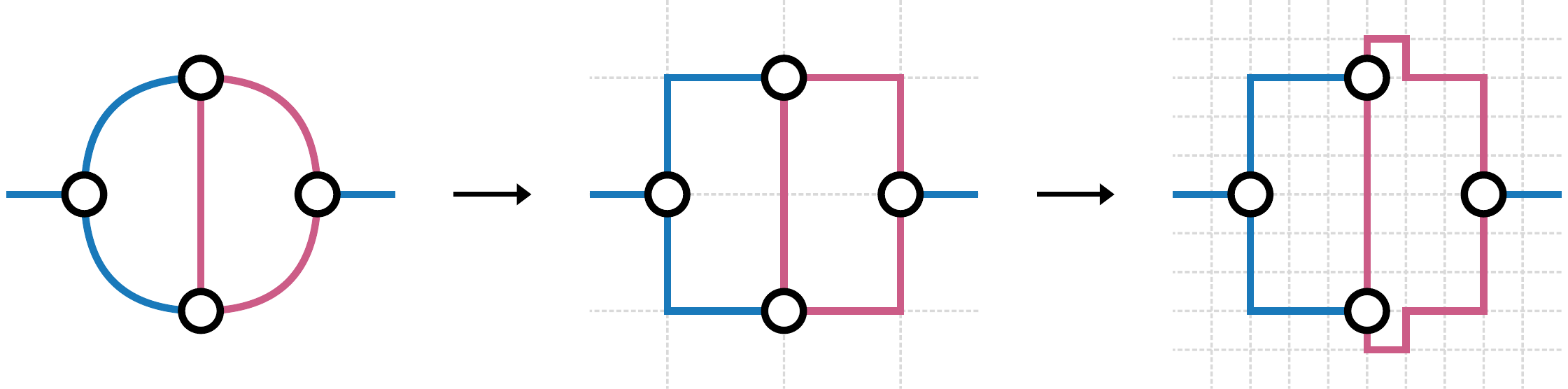}
\caption{A grid drawing of the graph in Step (1). Two red edges incident to an AND vertex are collinear.}
\label{fig:col}
\end{figure}

\paragraph{Step (1)}
We first compute a rectilinear embedding of $G$\footnote{A \emph{rectilinear embedding} of a graph $G$ is an embedding of $G$ to square grid
such that (1) each vertex is placed on a grid point, (2) each edge joins two grid points by a sequence of straight orthogonal line segments of
integer length with turning points on grid points, (3) no pair of vertices and turning points shares the same grid point.
See \cite{MS98} for the further details.},
which can be done in $O(n^2)$ time \cite{MS98}.
To simplify the construction, we bend some edges (by refining the grid) to make
two red edges incident to an AND vertex collinear (\figurename~\ref{fig:col}).
We then define the frame $F$ so that each \emph{cell} of size $13\times 13$ contains one of 
an AND vertex, an OR vertex, a unit straight segment of an edge (\emph{wire}), 
a unit turning segment of an edge (\emph{corner}),
and an empty cell.

\begin{figure}\centering
\includegraphics[width=0.8\linewidth]{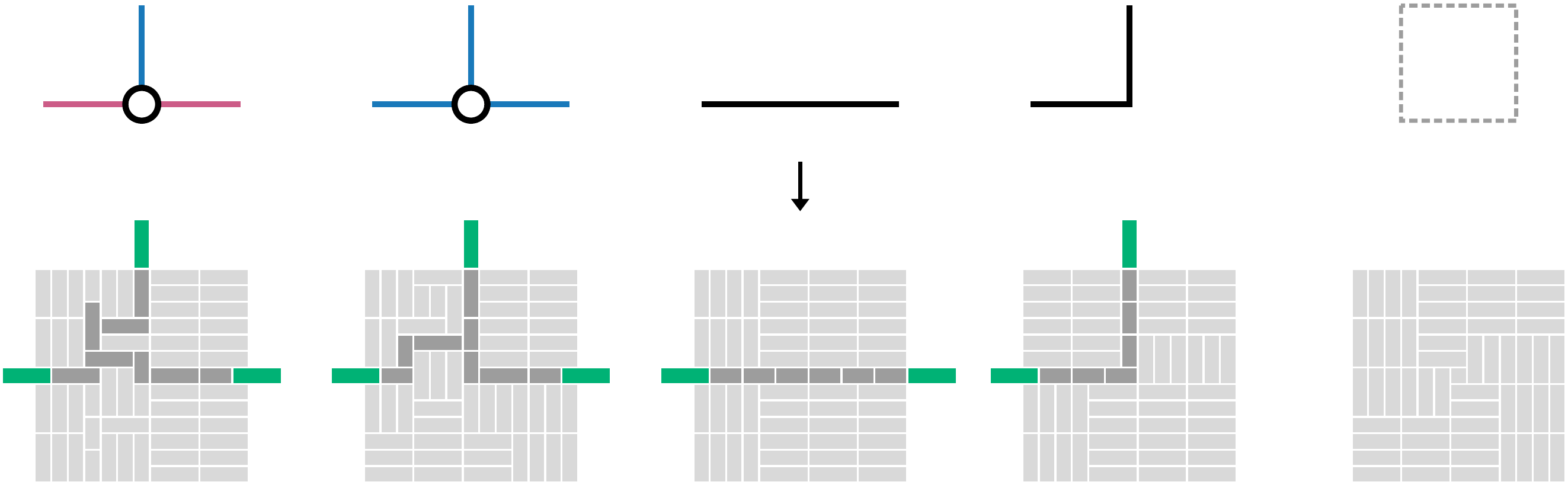}
\caption{Gadgets of size $13\times 13$ for an AND vertex, an OR vertex, a wire, a corner, and an empty cell.}
\label{fig:gadget}
\end{figure}


\paragraph{Step (2)}
The vertices, unit segments of edges, and empty cells are replaced by the gadgets shown in \figurename~\ref{fig:gadget}.
We note that these gadgets are designed in a square of size $13\times 13$
and green tetrmominoes are shared by both gadgets.
We embed these gadgets into the frame $F$ to represent the graph, with margins of gap 2.
If a vacant space of size $1\times 2$ remains at each place between two gadget, we fill it by a domino.

\begin{figure}\centering
\includegraphics[width=0.65\linewidth]{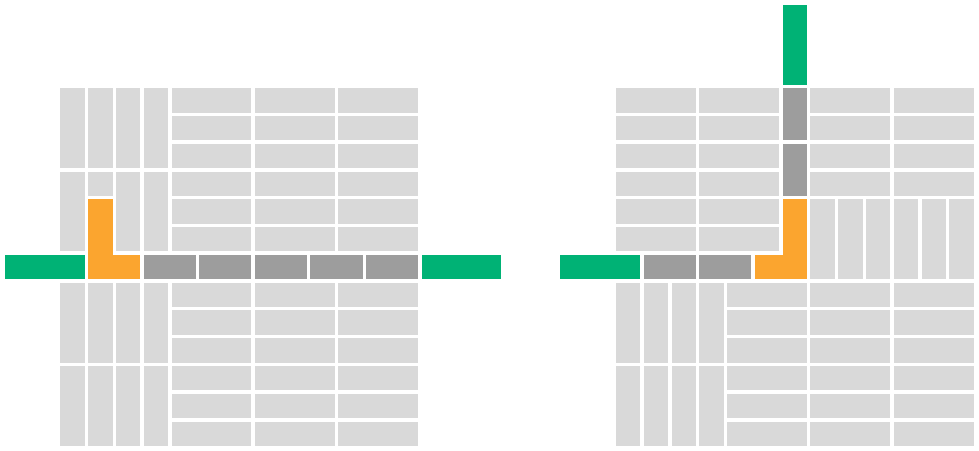}
\caption{Two edge gadgets for $e_t$ for embedding the L-tetromino.}
\label{fig:L-tetromino}
\end{figure}

\begin{figure}\centering
\includegraphics[width=0.8\linewidth]{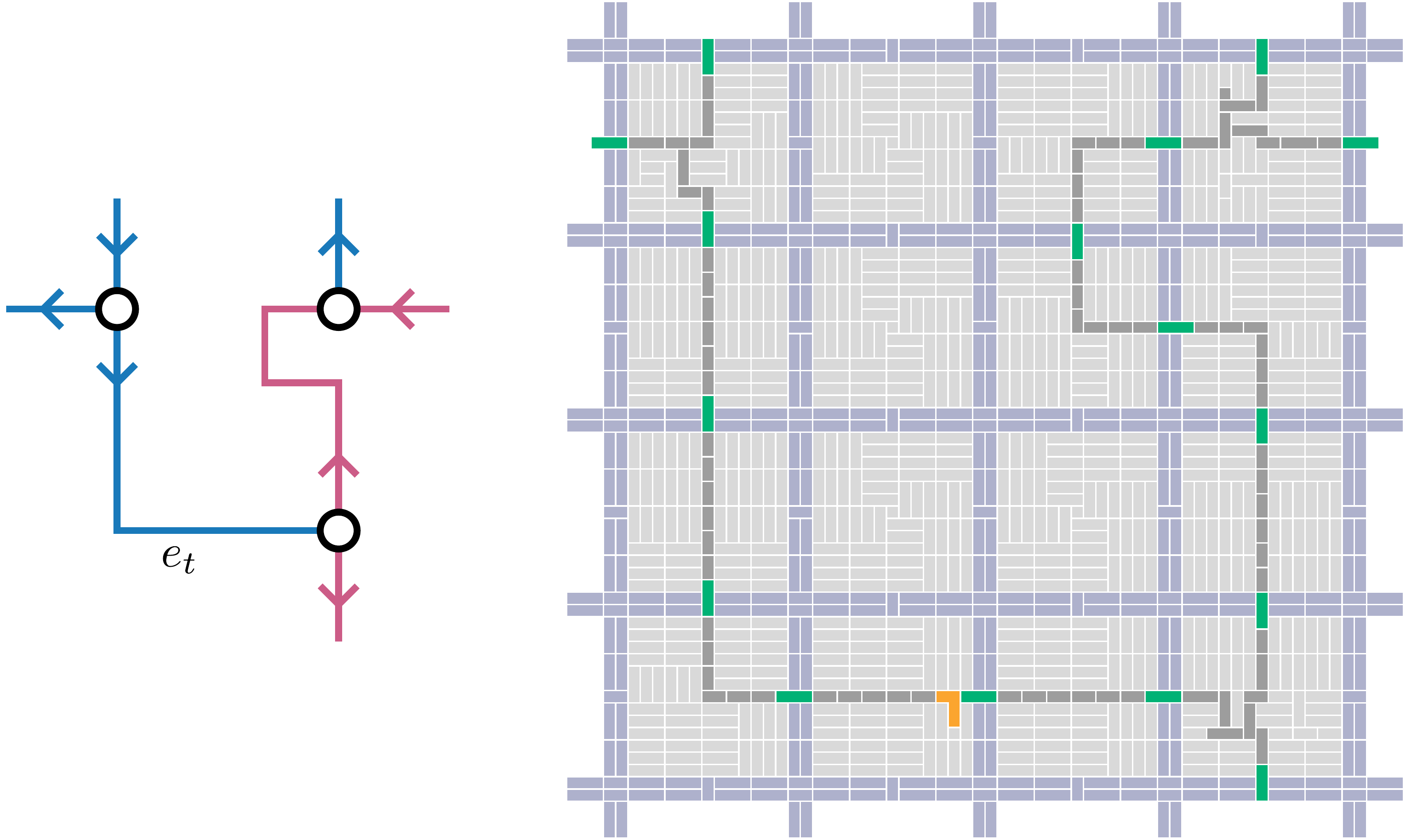}
\caption{Completed reduction after Step (3).}
\label{fig:complete}
\end{figure}

\paragraph{Step (3)}
Pick up any unit segment of the edge $e_t$, and embed
the left gadget in \figurename~\ref{fig:L-tetromino} if it is a straight, and
the right gadget in \figurename~\ref{fig:L-tetromino} if it is a turn.
We note that each of them contains the unique asymmetric polyomino of size 4 in L-shape (which is called the \emph{L-tetromino}),
and this L-tetromino can be flipped when the leftmost I-tromino moves to left.
Then we put polyominoes of unit size (which are called \emph{monominoes}) to fill all vacant unit squares except one.
(This one vacant unit square can be arbitrary.)
This is the end of our reduction.
An example of the construction is given in \figurename~\ref{fig:complete}.

\begin{figure}\centering
\includegraphics[width=0.4\linewidth]{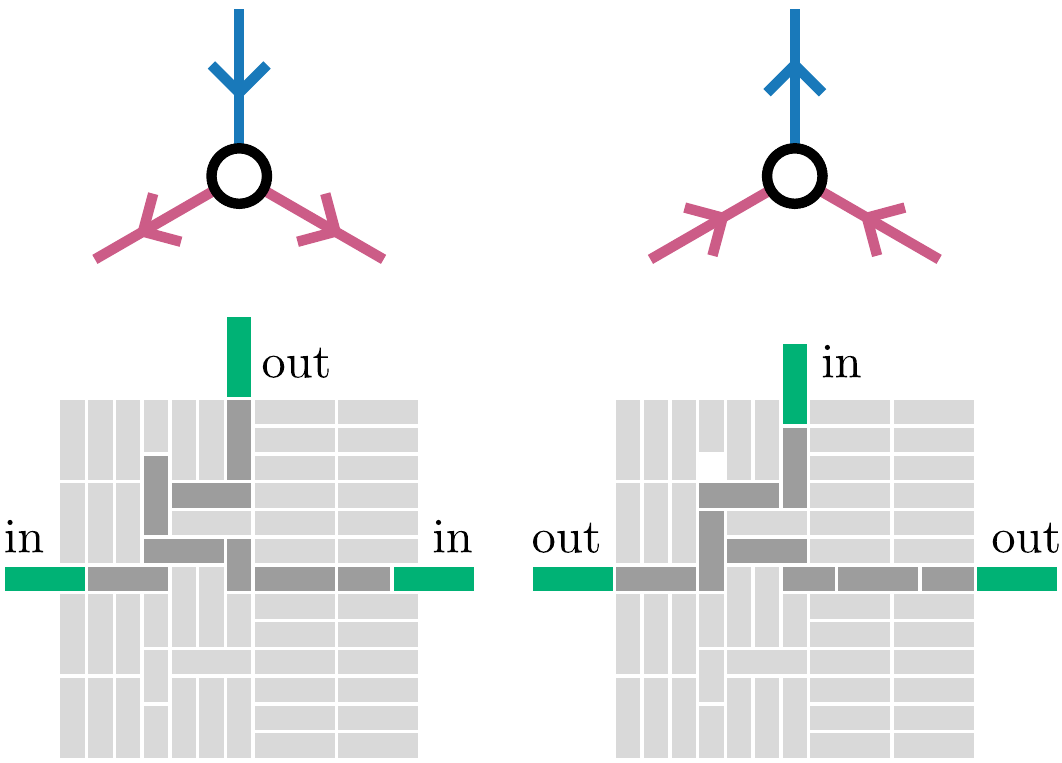}
\hspace{1cm}
\includegraphics[width=0.4\linewidth]{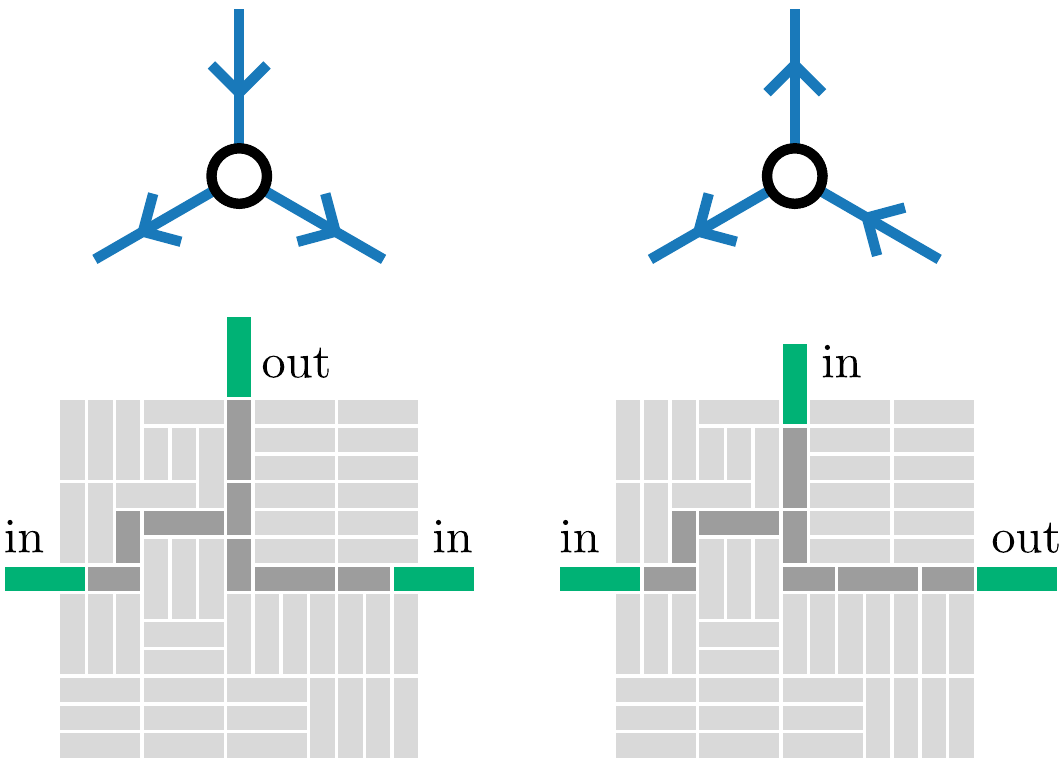}
\caption{Movements of gadgets.}
\label{fig:gadgets}
\end{figure}

It is clear that this reduction can be done in polynomial time,
the flip over puzzle constructed from $G$ satisfies the conditions (1) to (4) in Theorem \ref{th:main},
and the flip over puzzle is in PSPACE.
We show that $G$ is a yes instance of NCL if and only if
the instance of the flip over puzzle constructed from $G$ is a yes instance.
The basic idea is the same as one used in \cite[Sec.~9.3]{HearnDemaine2009};
when an edge $e$ changes the direction and points to a vertex $v$,
the corresponding \emph{I-tromino} (of size $1\times 3$ and in green in the figure) moves out one unit.
As shown in \figurename~\ref{fig:gadgets}, we designed two gadgets so that
(1) when both of the I-trominoes corresponding to the red edges moved out from an AND gadget,
the I-tromino corresponding to the blue edge of the AND gadget can move in, and 
(2) when one of the I-trominoes corresponding to the blue edges moved out from an OR gadget,
the I-tromino corresponding to the blue edge of the OR gadget can move in.
Using the same arguments in \cite[Sec.~9.3]{HearnDemaine2009},
we can show that the flip over puzzle simulates the movements of the NCL and vice versa.
The details are omitted here.

We note that the instance of the flip over puzzle constructed here has only one vacant unit square.
However, by jumping monominoes,
we can move the vacant unit square to anywhere we need as long as it is occupied by a monomino.
Thus, if we can eventually reverse $e_t$, we can make a vacant unit square next to the L-tetromino,
and then we can make it back-side up.
\qed
\end{proof}

It is known that NCL is still PSPACE-complete even if the input NCL graph is planar and have constant bandwidth
and it is given with a rectilinear embedding of constant height (see \cite[Thm.~11]{Zanden} for the details).
Thus we have the following corollary.
\begin{corollary}
The flip over puzzle is PSPACE-complete even with the conditions in Theorem \ref{th:main},
and the frame is of constant width.
\end{corollary}

%% file: np.tex
In this section, we give a proof of Theorem \ref{th:2np}.
We reduce the following \textsc{3-Partition} problem to our problem.

\ptitle{3-Partition}
\begin{listing}{aaa}
\item[{\bf Input:}] Positive integers $a_1,a_2,a_3,\ldots,a_{3m}$ such that
	     $\sum_{i=1}^{3m}a_i=mB$ for some positive integer $B$ and $B/4<a_i<B/2$ for $1\le i\le 3m$.
\item[{\bf Output:}] Determine whether we can partition $\{1,2,\ldots,3m\}$ into 
	     $m$ subsets $A_1,A_2,\ldots,A_m$ so that $\sum_{i\in A_j} a_i=B$ for $1\le j\le m$.
\end{listing}
\medskip

\begin{figure}\centering
\includegraphics[width=0.9\linewidth]{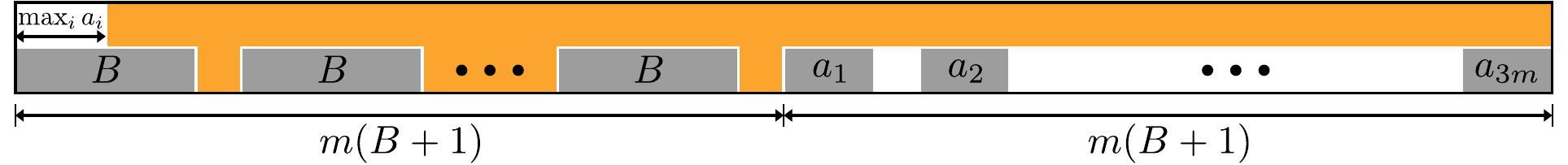}
\caption{Reduction for NP-hardness.}
\label{fig:3part}
\end{figure}

It is well known that the \textsc{3-Partition} problem is strongly NP-complete \cite{GJ79}.
Without loss of generality, we assume that $m<B/2$ (otherwise, we multiply each $a_i$ by $m'$ for some $m'>m$).
For a given instance $\langle a_1,a_2,\ldots,a_{3m}\rangle $ of \textsc{3-Partition},
we construct an instance of the jumping block puzzle as follows.
It consists of a frame $F$ and a set of pieces $\calP=\{P_1,P_2,\ldots,P_{4m+1}\}$.
The frame $F$ is a $2\times (2m(B+1))$ rectangle.
For $1 \le i \le 3m$, the piece $P_{i}$ is a $1 \times a_{i}$ rectangle.
The pieces $P_{3m+1}, P_{3m+2}, \dots P_{4m}$ are $1 \times B$ rectangles.
The unique asymmetric piece $P_{4m+1}$ is drawn in \figurename~\ref{fig:3part}.
It almost covers the whole frame $F$ except (1) one rectangle of size $1\times \max_i a_i$,
(2) $m$ rectangles of size $1\times B$ in the left side, and (3) one rectangle of size $1\times m(B+1)$ in the right side.
The initial feasible packing of $\calP$ to $F$ is also described in \figurename~\ref{fig:3part}.
All pieces are placed front-side up.
The asymmetric piece $P_{4m+1}$ is placed in $F$ as shown in \figurename~\ref{fig:3part},
and $m$ rectangles of size $1\times B$ are occupied by $m$ pieces $P_{3m+1},P_{3m+2},\ldots, P_{4m}$ of size $1\times B$.
The other $3m$ polyominoes $P_{1},P_{2},\ldots, P_{3m}$ are put
in the rectangle of size $1\times m(B+1)$ in the right side of $P_{4m+1}$ in arbitrary order.
We assume that the piece $P_{3m+i}$ occupies the $i$-th vacant space of size $1\times B$ from the center of $P_{4m+1}$.

Now, we show that all pieces in $\calP$ can be flipped back-side up if and only if
the original instance of \textsc{3-Partition} is a yes instance.
Since all pieces in $\calP$ are rectangles except $P_{4m+1}$,
it is sufficient to ask whether $P_{4m+1}$ can be flipped.

We first assume that all pieces in $\calP$ can be flipped.
It is easy to see that $P_{4m+1}$ can be flipped only in the horizontal direction.
We first observe that no piece of size $1\times B$ can be moved to any other place since
the vacant space at the top-left corner is small and we cannot make
a space for size $1\times B$ in the right-hand side even if the largest piece of size $1\times a_i$ is moved to
the top-left corner by the assumption $m<B/2$.
Now we consider the situation just before the asymmetric piece $P_{4m+1}$ is flipped.
%
The top-left vacant space of size $1\times \max_i a_i$ flips to the top-right corner, and
the corresponding top-right part of $P_{4m+1}$ flips to the top-left corner.
Therefore, there is no piece in this vacant space when we flip $P_{4m+1}$.
In the bottom-left half in $P_{4m+1}$, we have $m$ vacant spaces of size $1\times B$,
and each space is filled by a piece $P_i$ of size $1\times B$ for $3m+1\le i\le 4m$.
When the piece $P_{4m+1}$ is flipped, all these $m$ vacant spaces are flipped to the right half of the frame.
On the other hand, the total length of the remaining pieces $P_1,P_2,\ldots,P_{3m}$ is $mB$.
They should be in the right half of the frame, and $P_{4m+1}$ cannot overlap with any of
them when it is flipped.
Therefore, these pieces $P_1,P_2,\ldots,P_{3m}$ should fit into $m$ vacant spaces of size $1\times B$ in
$P_{4m+1}$ when these vacant spaces come from the left side to the right side by the flip of $P_{4m+1}$.
Therefore, all of these $m$ vacant spaces of size $1 \times B$ in $P_{4m+1}$ should be filled by
the pieces $P_1,P_2,\ldots,P_{3m}$ when $P_{4m+1}$ is flipped.
In order to do that, $P_1,P_2,\dots,P_{3m}$ should be partitioned into $m$ groups so that
each group exactly fills a vacant space of size $1\times B$ in total.
This partition clearly gives us a solution of the instance of the \textsc{3-Partition} problem.

We next assume that the instance of the \textsc{3-Partition} problem is a yes instance.
Then $\{1,\dots,3m\}$ has a partition into $m$ 3-subsets $A_1,A_2,\ldots,A_m$,
where each $A_i$ consists of $i_1,i_2,i_3$ with $a_{i_1}+a_{i_2}+a_{i_3}=B$.
Using the left-top vacant space as a working space,
we can sort the pieces $P_{1}, \dots, P_{m}$ into \emph{any order} in the right-bottom space.
(E.g., we can simulate the bubble sort.)
Thus we sort $P_{1}, \dots, P_{m}$ in such a way that for $1 \le i \le m$,
the positions $3i-2, 3i-1, 3i$ are occupied by the pieces $P_{i_{1}},P_{i_{2}},P_{i_{3}}$.
Finally, we adjust the positions of the pieces to have a vacant space of unit size after every third piece.
After this process, since we use a solution of \textsc{3-Partition},
$P_1,\ldots,P_{3m}$ form $m$ rectangles of size $1\times B$. Therefore we can flip $P_{4m+1}$.
This completes the proof of Theorem \ref{th:2np}.
\qed

%% file: np2.tex
In the proof of Theorem \ref{th:2np}, we reduced the \textsc{3-Partition} problem to the flip over puzzle.
We saw that the instances constructed there had polynomial-length yes-witnesses,
however, we do not know whether the flip over puzzle on a frame of width 2 is in NP or not.
In this section, we focus on the flip over puzzle with constant flips.
In this model, we restrict ourselves that the number of flips for each piece is bounded above by a constant.
We assume that each piece can be moved at most $c$ times for some constant $c$.
Then, if an instance of the puzzle has a solution, the solution can be represented by a sequence of moves,
and the number of moves is bounded above by $c n$, where $n$ is the number of pieces.
Therefore, each yes-instance has a witness of polynomial length, which implies that the puzzle with the restriction is in NP.
In this section, we show that the flip over puzzle is still NP-complete even if each piece can be flipped at most once
with some additional restrictions. 

\begin{theorem}
\label{th:npc1}
The flip over puzzle with constant flips is NP-complete even if it satisfies all the following conditions (0), (1), (2), and (3).
(0) Each piece can be moved at most once, 
(1) the frame is a rectangle without holes,
(2) every line symmetric piece is a rectangle of size $1\times k$ with $1\le k\le 3$, and
(3) there is only one asymmetric piece of size 4.
\end{theorem}

\begin{figure}\centering
\includegraphics[width=0.4\linewidth]{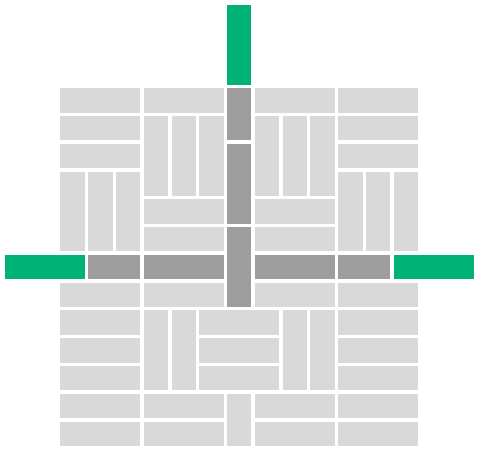}
\caption{AND gadget for Theorem \ref{th:npc1}.}
\label{fig:AND}
\end{figure}

\begin{proof}
The claim can be proved by following the same strategy of the proof of Theorem \ref{th:main} in Section \ref{sec:pspace}.
That is, we give a polynomial-time reduction from the bounded NCL problem.
The construction from the NCL graph can be done in almost the same way shown in Section \ref{sec:pspace} except Step (3) and the AND gadget.
This time, we do not fill each of vacant unit squares by monominoes in Step (3).
In the reduction in the proof, the monominoes are only pieces that need to be moved twice or more.
Thus the resultant puzzle satisfies the condition (0).
On the other hand, if we remove all monominoes from the gadgets,
the AND gadget in \figurename~\ref{fig:gadget} can produce an empty domino, and hence it does not work as is.
To avoid it, we replace the AND gadget in \figurename~\ref{fig:gadget} by
one given in \figurename~\ref{fig:AND}.\footnote{We note that the AND gadget in \figurename~\ref{fig:AND}
does not work in the proof of Theorem \ref{th:main} since this gadget requires at least two vacant unit squares.}
Then we can observe that an instance of the bounded NCL problem is a yes instance
if and only if the constructed instance of the flip over puzzle has a solution with the conditions (0), (1), (2), and (3).
\qed\end{proof}

\begin{theorem}
\label{th:npc2}
The flip over puzzle with constant flips is NP-complete even if it satisfies all the following conditions (0), (1), (2), (3), and (4).
(0) Each piece can be moved at most once, 
(1) the frame is a rectangle without holes,
(2) every line symmetric piece is a rectangle of size $1\times k$ with $1\le k\le 3$,
(3) all asymmetric pieces are of size 4, and
(4) there is only one vacant unit square.
\end{theorem}

\begin{figure}\centering
\includegraphics[width=0.8\linewidth]{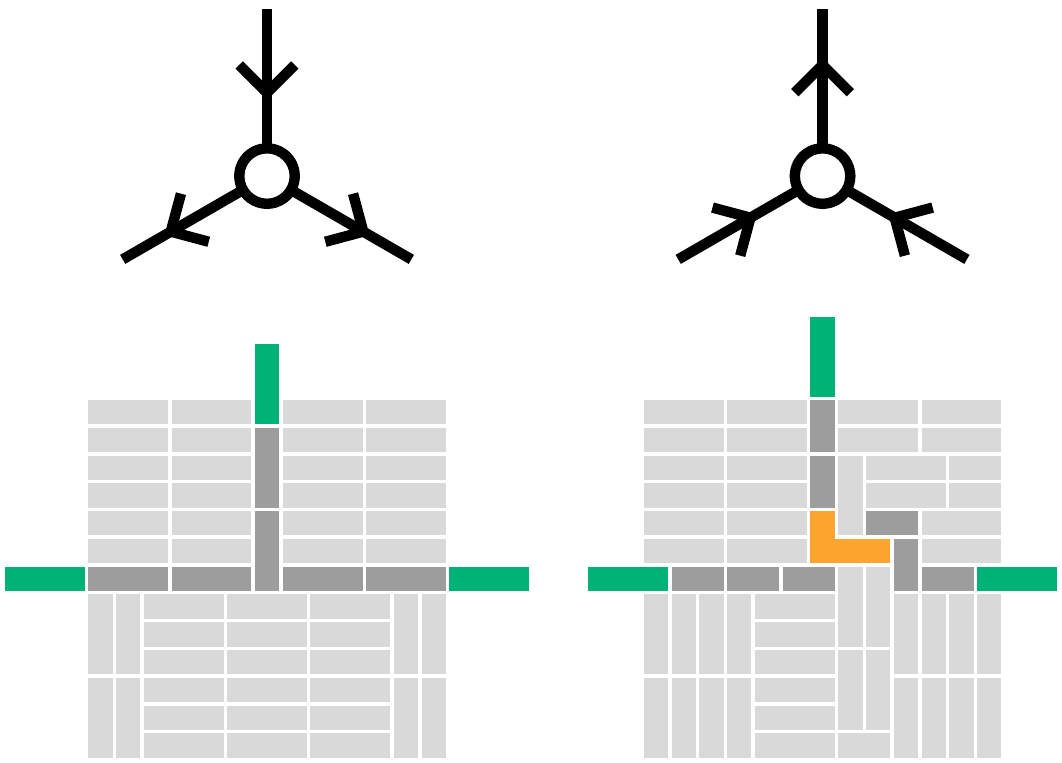}
\caption{The two gadgets for replacing the two types of vertices.}
\label{fig:two}
\end{figure}

\begin{proof}
In order to prove NP-hardness, we give a polynomial-time reduction from
the \emph{Hamiltonian cycle problem} on 3-regular planar digraphs where each vertex has indegree 1 or 2
which is a classic well-known NP-complete problem \cite{Ple79}.
We use an idea similar to the one in \cite{UI90}: For a given 3-regular planar digraph $G$,
we replace the two types of vertices by the two types of gadgets shown in \figurename~\ref{fig:two}, respectively,
and join them by the same wire gadget in Section \ref{sec:pspace}.
Our construction is done by making one (and unique) vacant unit square in a wire gadget. 
(We note that every edge joining an outdegree-1 vertex to an indegree-1 vertex should belong to any Hamiltonian cycle.\footnote{
Furthermore, the Hamiltonian cycle problem on 3-regular planar digraphs is still NP-complete even if
(1) the numbers of outdegree-1 vertices and indegree-1 vertices are the same, and
(2) each outdegree-1 vertex should be joined to an indegree-1 vertex.
In the other cases, we can reduce redundant edges that never be used in an Hamilton cycle.
The details are omitted here.} Therefore, the vacant unit square can be put in any edge of this type.)
The unique vacant unit square moves like the peg in \cite{UI90}: when it visits an indegree-1 vertex as shown in
the left side of \figurename~\ref{fig:two}, the central vertical I-trominoes in the figure
move the vacant unit square to center, and it will be brought to left or right.
On the other hand, when the vacant unit square visits an indegree-2 vertex from one of the two incoming directions,
then the L-tetromino can move it to the outgoing direction of the vertex. 
Therefore, $G$ has a Hamiltonian cycle if and only if the flip over puzzle has a solution with the conditions (0), (1), (2), and (3).
\qed\end{proof}

\begin{theorem}
\label{th:npc3}
The flip over puzzle with constant flips is NP-complete even with all the following conditions:
(0) Each piece can be moved at most once, 
(1) the frame is a rectangle of width $2$ without holes,
(2) all pieces are rectangles of width 1 except 2, and
(3) there is only one asymmetric piece.
\end{theorem}

\begin{figure}\centering
\includegraphics[width=0.9\linewidth]{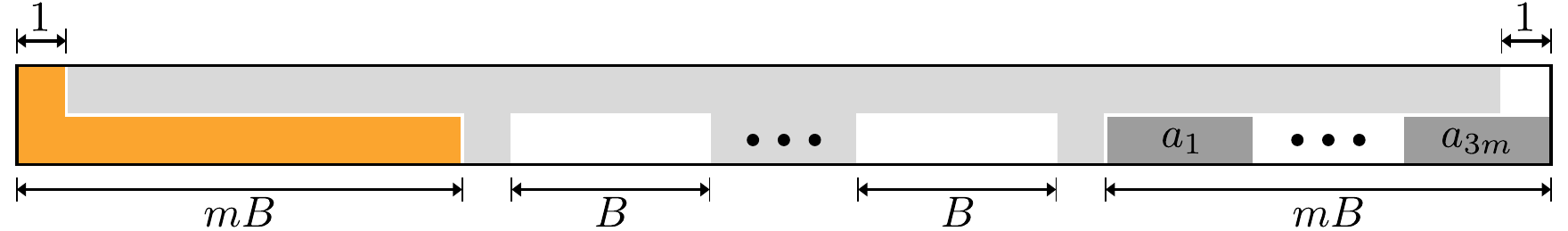}
\caption{Gadgets for the flip over puzzle of width $2$.}
\label{fig:flip2}
\end{figure}

\begin{proof}
The basic idea is similar to the proof of Theorem \ref{th:2np}.
We reduce the \textsc{3-Partition} problem to our problem as shown in \figurename~\ref{fig:flip2}.
For given positive integers $a_1,a_2,a_3,\ldots,a_{3m}$, $F$ is a rectangle of size $2\times (3mB+m+1)$.
Each $P_i$ is a rectangle of size $1\times a_i$ for $1\le i\le 3m$.
The large piece $P_{3m+1}$ is as shown in \figurename~\ref{fig:flip2}:
$P_{3m+1}$ has $m$ vacant rectangles of size $1\times B$ in center,
and two vacant rectangles of size $1\times (mB-1)$ on the left and right sides.
The last piece $P_{3m+2}$ is a long L-shape polyomino as shown in \figurename~\ref{fig:flip2}.
Since each block can be flipped at most once, the only way to flip $P_{3m+2}$ is
(0) flip $P_{3m+1}$, 
(1) flip $P_i$ for each $i$ with $1\le i\le 3m$ into $m$ vacant rectangles of $1\times B$ in center of $P_{3m+1}$, and
(2) flip $P_{3m+2}$ from the left side in $F$ to the right side in $F$.
Therefore, the instance of the \textsc{3-Partition} problem is a yes instance if and only if
the constructed instance of the flip over puzzle with the conditions has a solution.
\qed\end{proof}

%% file: fly.tex
In this section, we turn to the flying block puzzles.
In flying block puzzles, we can translate and rotate pieces but cannot flip them.
In the flip over puzzle in Section \ref{sec:flip}, almost all pieces are rectangles except a few asymmetric pieces.
Since a rectangle does not change by a flip, we can inherit most of the results for flip over puzzles in Section \ref{sec:flip}
to ones for flying block puzzles by changing some special pieces.
However, while the goal of the flip over puzzle is to flip \emph{all} pieces,
the goal of the flying block puzzle is to arrange a \emph{specific} piece to the goal position.
This difference requires different techniques.

The counterpart of Theorem \ref{th:main} is as follows:
\begin{theorem}
\label{th:main-fly}
The flying block puzzle is PSPACE-complete even with all the following conditions:
(1) the frame is a rectangle without holes,
(2) every piece is a rectangle of size $1\times k$ with $1\le k\le 3$, and
(3) there is only one vacant unit square.
\end{theorem}
\begin{proof}(Outline)
In the proof of Theorem~\ref{th:main}, we use an L-tetromino as the unique asymmetric piece.
In the flying block puzzle, we do not need to use this trick.
Instead of that, the goal of the resulting flying block puzzle is just to move
any specific I-tromino in the wire or corner gadget in \figurename~\ref{fig:gadget} corresponding to $e_t$,
which plays the same role of the flip of the L-tetromino in the proof of Theorem~\ref{th:main}.
\qed\end{proof}

The counterpart of Theorems \ref{th:2np} and \ref{th:npc3} is as follows:
\begin{theorem}
\label{th:npc3-fly}
The flying block puzzle is NP-complete even with all the following conditions:
(1) the frame is a rectangle of width $2$ without holes,
(2) every line symmetric piece is a rectangle of width 1, and
(3) there is only one non-rectangle piece, which is line symmetric.
\end{theorem}

\begin{figure}\centering
\includegraphics[width=0.9\linewidth]{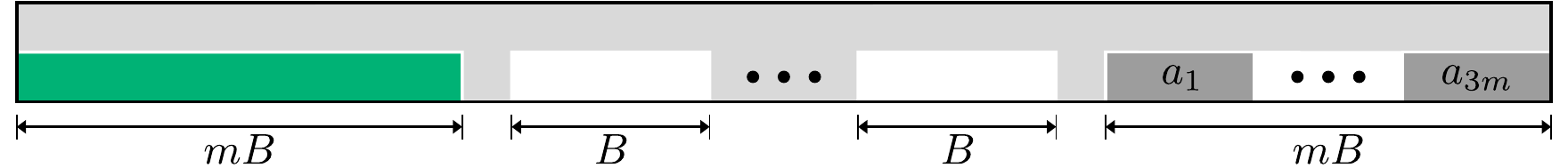}
\caption{Reduction to the flying block puzzle.}
\label{fig:flying2}
\end{figure}

\begin{proof}
The basic idea is the same as the proof of Theorem \ref{th:2np}: We reduce from the \textsc{3-Partition} problem.
The construction of an instance of the flying block puzzle from a given instance of the \textsc{3-Partition} problem is
given in \figurename~\ref{fig:flying2}.
The goal of this puzzle is to move the rectangle piece of size $1\times mB$ to the right side.
To do that, all other pieces of size $1\times a_i$ should be moved to the $m$ vacant spaces of size $1\times B$.
Therefore, the instance of the \textsc{3-Partition} problem is a yes instance if and only if
the constructed instance of the flying block puzzle has a solution.
\qed\end{proof}

We note that the claim in the proof of Theorem \ref{th:npc3-fly} still holds even if
each piece can be moved at most once. In this sense, Theorem \ref{th:npc3-fly} is the counterpart of
both Theorems \ref{th:2np} and \ref{th:npc3}.




The counterpart of Theorem \ref{th:npc1} is as follows:
\begin{theorem}
\label{th:npc1-fly}
The flying block puzzle with constant flies is NP-complete even if it satisfies the following conditions (0), (1), (2), and (3).
(0) Each piece can be moved at most once, 
(1) the frame is a rectangle without holes, and
(2) every piece is a rectangle of size $1\times k$ with $1\le k\le 3$.
\end{theorem}
\begin{proof}(Outline)
We apply the same technique of the proof of Theorem \ref{th:main-fly} to Theorem \ref{th:npc1}.
\qed\end{proof}

Interestingly, the reduction of the proof of Theorem \ref{th:npc2} from the Hamiltonian cycle
cannot be extended to the flying block puzzle since we do not need to
``visit'' every piece to flip. 
Therefore, the computational complexity of
the flying block puzzle with the conditions in Theorem \ref{th:npc2} is open.
We now turn to the flying block puzzles with the frame of width 1.
While this case is trivially ``yes'' in any instance of the flip over puzzle since every piece is a rectangle of width 1,
we show weak NP-hardness of the flying block puzzle if the frame is of width 1 and each piece can be
moved at most once. On the other hand, when we can move the pieces any number of times,
we have nontrivial polynomial time algorithm to solve it.

\begin{theorem}
\label{th:weakNP}
The flying block puzzle with constant flies is weakly NP-complete even if 
each piece can be moved at most once and the frame is a rectangle of width $1$.
\end{theorem}

\begin{figure}
\centering
\includegraphics[width=0.8\linewidth]{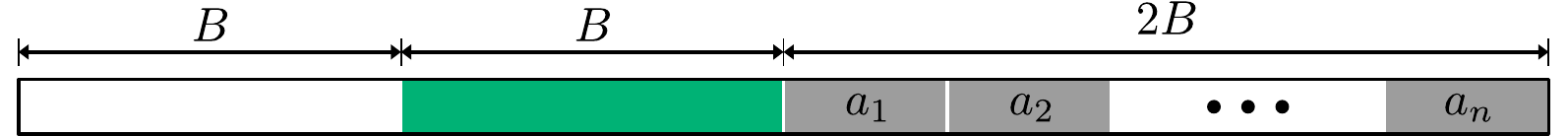}
\caption{The reduction from \textsc{Partition}.}
\label{fig:partition}
\end{figure}

\begin{proof}
We present a polynomial-time reduction from \textsc{Partition}~\cite{GJ79},
which asks, given $n$ positive integers $a_{1}, \dots, a_{n}$ with $\sum_{i=1}^{n} a_{i} = 2B$,
whether there exists a partition of $\{1,\dots,n\}$ into two sets $A$ and $A'$
such that $\sum_{i \in A} a_{i} = \sum_{i \in A'} a_{i} = B$.
This problem is weakly NP-complete. Namely, the problem is NP-complete if each number is given as its binary representation.

For a given instance $\langle a_{1},a_{2},\dots,a_{n}\rangle $ of \textsc{Partition},
we construct an instance of the flying block puzzle as follows.
It consists of a frame $F$ and a set of pieces $\calP = \{P_{1}, P_{2},\dots,P_{n+1}\}$.
The frame $F$ is a $1\times (4B)$ rectangle.
For $1 \le i \le n$, the piece $P_{i}$ is a $1 \times a_{i}$ rectangle.
The special piece $P_{n+1}$ is a $1 \times B$ rectangle.
The initial feasible packing of $\calP$ to $F$ is given in \figurename~\ref{fig:partition}.
The goal is to move $P_{n+1}$ to the right end of $F$.

First assume that $\langle a_{1},a_{2},\dots,a_{n}\rangle $ is a yes-instance of \textsc{Partition}.
That is, there is a partition of $\{1,\dots,n\}$ into $A$ and $A'$ such that $\sum_{i \in A} a_{i} = \sum_{i \in A'} a_{i} = B$.
We first move the pieces $P_{i}$ with $i \in A$ to the leftmost vacant space of size $1\times B$
and then we move and pack the pieces $P_{i}$ with $i \in A'$ to left as possible as we can in the vacant space of size $1\times 2B$.
Now we can move $P_{n+1}$ to the goal position.
Note that each piece moves at most once in this process.

Next assume that there exists a sequence of moves that eventually places $P_{n+1}$ to the goal position.
Let us consider the step right before $P_{n+1}$ moves to the goal position in $F$ for the first time.
We set $A$ to the indices of the pieces placed in the left side of $P_{n+1}$ in $F$ at this moment and $A'$ to $\{1,\dots,n\} \setminus A$.
At this moment, there has to be a vacant space of $1 \times B$ rectangle in the goal position.
Therefore, $\sum_{i \in A'} a_{i} \le B$ holds.
On the other hand, since $P_{n+1}$ is in the initial position now, $\sum_{i \in A} a_{i} \le B$ holds too.
Since the sum of $\sum_{i \in A} a_{i}$ and $\sum_{i \in A'} a_{i}$ is $2B$, both of them have to be exactly $B$.
\qed
\end{proof}


\begin{theorem}
\label{th:strong-poly}
The flying block puzzle can be solved in polynomial time when the frame is a rectangle of width $1$.
\end{theorem}

\begin{proof}
Let $F$ be the frame of size $1\times f$, and $\{Q,P_1,\dots,P_n\}$ be the set of pieces,
where $Q$ is the special piece to be moved to the given goal position.
Let $a_{i}$ denote the length of $P_{i}$ for each $i$.
We assume without loss of generality that $a_{1} \le a_{2} \le \cdots \le a_{n}$.
Let $a_{0}$ be the length of $Q$.
Now let $S=f-\sum_{i=0}^n a_i$, which gives us the total vacant space in $F$.
Using $S$, we divide the set of pieces into two groups.
A piece $P_i$ (or $Q$) is \emph{short} if $a_i\le S$, and \emph{long} if $a_i>S$.
We assume that the pieces $P_1,\dots,P_{n'}$ are short, and
$P_{n'+1},\dots,P_n$ are long. Let $\mathcal{X}$ be the set of long pieces.
(The set $\mathcal{X}$ may contain $Q$.)

Observe that two long pieces cannot exchange their relative positions in $F$
since the working space $S$ is too small to swap them.
Therefore, the ordering in $\mathcal{X}$ cannot be changed by any flies of pieces.
On the other hand, when at least one of them is short,
two pieces can exchange their relative positions in $F$ in linear number of flies
without changing the relative positions of the other pieces.

Now we focus on the special piece $Q$.
If $Q$ is in $\mathcal{X}$, the relative position of $Q$ in $\mathcal{X}$ is fixed.
On the other hand, if $Q$ is short, the number of possible orderings
of $\mathcal{X} \cup \{Q\}$ is $\msize{\mathcal{X}}+1$.
In the latter case, the algorithm checks all of these $\msize{\mathcal{X}}+1$ cases.
Hereafter, we assume that the ordering of $\mathcal{X}\cup \{Q\}$ is fixed.
(If $Q$ is short and not in that position in $\mathcal{X}\cup \{Q\}$,
then we first move $Q$ to that relative position.)

When we put $Q$ at the goal position, the frame $F$ is divided into
two parts by $Q$.
let $f_L$ and $f_R$ be the length of the left and right parts of $F$, respectively.
In the ordering of $\mathcal{X}\cup \{Q\}$, let $\mathcal{X}_L$ and
$\mathcal{X}_R$ be the sets of long pieces at the left and right of $Q$ in $F$, respectively.
Let $L=f_L-\sum_{P_i\in \mathcal{X}_L} a_i$ and $R=f_R-\sum_{P_i\in \mathcal{X}_R}a_i$,
which give us the total spaces we can use for placing short pieces in
the left and right sides in $F$ divided by $Q$ at the goal position.

If we have already $L<0$ or $R<0$, the answer is ``no''.
On the other hand, if $L\ge 0$ and $R\ge 0$, we can decide that the
answer is ``yes'' as follows.

Now, we have to construct a subset $A\subseteq\{1,2,\dots,n'\}$
so that $\sum_{i\in A} a_i\le L$ and $\sum_{i\not\in A} a_i\le R$.
This set $A$ can always be constructed as follows.
We greedily add as many elements $P_i$ in $\{P_1,P_2,\ldots,P_{n'}\}$
as possible to $A$ while keeping $\sum_{i\in A} a_i \le L$.
Then, we have $L-\sum_{i\in A} a_i<S$ since $a_i \le S$ for every short piece $P_i$.
Therefore, we have $L-S<\sum_{i\in A}a_i\le L$.
On the other hand, $\sum_{i\in A} a_i+ \sum_{i\not\in A} a_i=\sum_{1\le i\le n'}a_i = L+R-S$.
Thus we obtain $\sum_{i\not\in A}a_i=L+R-S-\sum_{i\in A} a_i<L+R-S-(L-S)=R$.
\qed
\end{proof}

%% file: conc.tex
In this paper, we investigate the computational complexities of the jumping block puzzles which
form the token-jumping counterpart of the sliding block puzzles in the context of reconfiguration problems.
The other well-studied model in the field of combinatorial reconfiguration would allow
``removals and additions'' of blocks, which would be future work.
Even in the jumping block puzzles, we still have many variants.

One natural variant in the context of the combinatorial reconfiguration is that
the input consists of the initial feasible packing and the target feasible packing.
In this reconfiguration problem, we have two observations (proofs are omitted):
\begin{obs}
The reconfiguration problem of the jumping block puzzle is tractable if the frame $F$ is of width $1$.
\end{obs}
\begin{obs}
The reconfiguration problem of the jumping block puzzle is tractable if all pieces are rectangles of size $1\times 2$.
\end{obs}
Extension of them to general cases, e.g., blocks of size $1 \times 3$, seems to be interesting (cf.~\cite{HINSU2012}).

\begin{figure}
\centering
\includegraphics[width=0.2\linewidth]{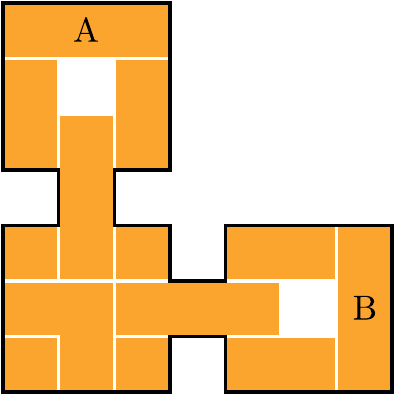}
\caption{A ``simple'' example of the flying block puzzle designed by Hideyuki Ueyama in \cite{Akiyama} (with permissions):
   It requires 256 steps to exchange A and B.}
\label{fig:usample}
\end{figure}

We may allow the frame to have holes (or fixed obstacles).
Moreover, we may distinguish congruent pieces.
One interesting example can be found in \cite{Akiyama}:
the puzzle in \figurename~\ref{fig:usample} was designed by Hideyuki Ueyama.
It consists of 12 rectangle pieces and one L-tromino.
There are two vacant unit squares. The goal of this puzzle is to exchange
two I-trominoes with labels A and B. At a glance, it seems to be impossible.
However, they can be exchanged in 256 steps.
In such a case, it seems that it requires exponential steps if we have a few vacant unit squares.
Intuitively, the puzzle is likely to become easier if we have many vacant unit squares.
(In fact, some concrete examples can be found in \cite{Akiyama}.)
It would be interesting to ask whether there is an algorithm
that runs faster when $k$ is larger, where $k$ is the number of vacant unit squares.

%% file: main.bbl
\begin{thebibliography}{10}
\providecommand{\url}[1]{\texttt{#1}}
\providecommand{\urlprefix}{URL }
\providecommand{\doi}[1]{https://doi.org/#1}

\bibitem{Akiyama}
Akiyama, H.: {Board Puzzle Reader (in Japanese)}. Shin Kigen Sha (2009)

\bibitem{BJLPP11}
Bonamy, M., Johnson, M., Lignos, I., Patel, V., Paulusma, D.: {On the diameter
  of reconfiguration graphs for vertex colourings}. Electronic Notes in
  Discrete Mathematics  \textbf{38},  161--166 (2011)

\bibitem{BC09}
Bonsma, P., Cereceda, L.: {Finding Paths between graph colourings:
  PSPACE-completeness and superpolynomial distances}. Theor. Comput. Sci.
  \textbf{410}(50),  5215--5226 (2009),
  \url{http://dx.doi.org/10.1016/j.tcs.2009.08.023}

\bibitem{CHJ11}
Cereceda, L., van~den Heuvel, J., Johnson, M.: {Finding paths between
  3-colourings}. J.~Graph Theory  \textbf{67},  69--82 (2011)

\bibitem{DemaineRudoy2018}
Demaine, E.D., Rudoy, M.: {A simple proof that the $(n^2-1)$-puzzle is hard}.
  Theoretical Computer Science  \textbf{732},  80--84 (2018)

\bibitem{GJ79}
Garey, M.R., Johnson, D.S.: {Computers and Intractability --- A Guide to the
  Theory of NP-Completeness}. Freeman (1979)

\bibitem{Gol}
Golomb, S.W.: {Polyominoes}. Princeton University Press (1994)

\bibitem{Kolaitis}
Gopalan, P., Kolaitis, P.G., Maneva, E.N., Papadimitriou, C.H.: {The
  connectivity of Boolean satisfiability: computational and structural
  dichotomies}. SIAM J.~Computing  \textbf{38},  2330--2355 (2009)

\bibitem{HearnDemaine2009}
Hearn, R.A., Demaine, E.D.: {Games, Puzzles, and Computation}. A K Peters Ltd.
  (2009)

\bibitem{HearnDemaine2005}
Hearn, R.A., Demaine, E.D.: {PSPACE-completeness of sliding-block puzzles and
  other problems through the nondeterministic constant logic model of
  computation}. Theoretical Computer Science  \textbf{343}(1-2),  72--96 (2005)

\bibitem{HINSU2012}
Horiyama, T., Ito, T., Nakatsuka, K., Suzuki, A., Uehara, R.: {Packing
  Trominoes is NP-Complete, {\#}P-hard and ASP-Complete}. In: {The 24th
  Canadian Conference on Computational Geometry (CCCG 2012)}. pp. 219--224
  (2012)

\bibitem{IDHPSUU}
Ito, T., Demaine, E.D., Harvey, N.J., Papadimitriou, C.H., Sideri, M., Uehara,
  R., Uno, Y.: {On the complexity of reconfiguration problems}. Theoretical
  Computer Science  \textbf{412},  1054--1065 (2011)

\bibitem{KMP11}
Kami\'nski, M., Medvedev, P., Milani\v{c}, M.: {Shortest paths between shortest
  paths}. Theoretical Computer Science  \textbf{412},  5205--5210 (2011)

\bibitem{KaminskiMedvedevMilanic2012}
Kami\'nski, M., Medvedev, P., Milani\v{c}, M.: {Complexity of independent set
  reconfigurability problems}. Theoretical Computer Science  \textbf{439},
  9--15 (2012)

\bibitem{MS98}
Liu, Y., Morgana, A., Simeone, B.: {A linear algorithm for 2-bend embeddings of
  planar graphs in the two-dimensional grid}. Discrete Applied Mathematics
  \textbf{81}(1),  69--91 (1998)

\bibitem{MTY11}
Makino, K., Tamaki, S., Yamamoto, M.: {An exact algorithm for the Boolean
  connectivity problem for $k$-CNF}. Theoretical Computer Science
  \textbf{412},  4613--4618 (2011)

\bibitem{Ple79}
Plesn\'{\i}k, J.: {The NP-completeness of the Hamiltonian Cycle Problem in
  Planar Digraphs with Degree Bound Two}. Information Processing Letters
  \textbf{8}(4),  199--201 (1979)

\bibitem{RW90}
Ratner, D., Warmuth, M.: {The $(n^2-1)$-puzzle and related relocation
  problems}. Journal of Symbolic Computation  \textbf{10},  111--137 (1990)

\bibitem{UI90}
Uehara, R., Iwata, S.: {Generalized Hi-Q is NP-Complete}. The Transactions of
  the IEICE  \textbf{E73}(2),  270--273 (1990), the proof is available at {\tt
  http://www.jaist.ac.jp/\verb+~+uehara/pdf/phd7.ps.gz}

\bibitem{Zanden}
van~der Zanden, T.C.: {Parameterized Complexity of Graph Constraint Logic}. In:
  {IPEC 2015}. pp. 282--293. LIPIcs Vol. 43, Dagsthul (2015).
  \doi{10.4230/LIPIcs.IPEC.2015.282}

\end{thebibliography}
